%% file: main.tex
\documentclass[12pt]{article}

\usepackage{amssymb,amsfonts,amsmath,amsthm,mathtools,bm}

\usepackage{bbm}

\usepackage[svgnames]{xcolor}

\let\SetColor\color

\usepackage{etoolbox}
\makeatletter
\AtBeginEnvironment{align*}{\let\SetColor\@gobble \color{black}}
\AtBeginEnvironment{align}{\let\SetColor\@gobble \color{black}}
\AtBeginEnvironment{aligned*}{\let\SetColor\@gobble \color{black}}
\AtBeginEnvironment{aligned}{\let\SetColor\@gobble \color{black}}
\AtBeginEnvironment{equation*}{\let\SetColor\@gobble \color{black}}
\AtBeginEnvironment{equation}{\let\SetColor\@gobble \color{black}}
\AtBeginEnvironment{underline}{\let\SetColor\@gobble \color{black}}
\AtBeginEnvironment{tabular}{\let\SetColor\@gobble \color{black}}
\makeatother

\definecolor{MyBlue}{RGB}{0,91,148}
\definecolor{MyRed}{RGB}{200,15,62}
\definecolor{MyPurple}{rgb}{0.44, 0.16, 0.39}

\usepackage{caption, subcaption,enumerate,color}
\usepackage[onehalfspacing]{setspace}
\usepackage{graphicx,multicol}
\usepackage[normalem]{ulem}
\usepackage{soul}
\usepackage{adjustbox}

\usepackage{tabu}
\usepackage{makecell,multirow}

\usepackage{titlesec}
  \titleformat*{\section}{\sc \centering \Large}
  \titleformat*{\subsection}{\sc \centering \large}
  \titleformat*{\subsubsection}{\sc \large}
  \titlelabel{\thetitle. }
  \titlespacing*{\section}
  {0pt}{1.5ex plus 0.3ex minus .1ex}{1.5ex plus .1ex}
  \titlespacing*{\subsection}
  {0pt}{1.5ex plus 0.3ex minus .1ex}{1.5ex plus .1ex}
  \titlespacing*{\subsubsection}
  {0pt}{0.75ex plus 0.15ex minus .05ex}{0.75ex plus 0.5ex}

\definecolor{UBCblue}{rgb}{0.0, 0.25, 0.75}

\usepackage{float}
\usepackage[
            colorlinks = true,
            linkcolor  = violet,
            citecolor  = DarkBlue,
            urlcolor   = UBCblue,
            ]{hyperref}

\newcommand{\sep}[0]{ \; |\; }

\newcommand{\supp}[0]{ \text{{\normalfont supp}\;}}

\usepackage[
            justification = centering,
            format        = plain,
            labelfont     = {bf,it},
            textfont      = it,
            belowskip     = -5pt,
            labelsep      = period
            ]{caption}

\usepackage[
            backend      = biber,
            style        = philosophy-classic,
            maxcitenames = 3,
            scauthorsbib = true,
            dashed       = true,
            doi          = false,
            isbn         = false,
            url          = false,
            eprint       = false,
            backref      = true,
            uniquename   = false
            ]{biblatex}

\renewbibmacro{in:}{}
\DeclareDelimFormat[cbx@textcite]{nameyeardelim}{\addspace}
\newcommand\citenop[1]{\citeauthor{#1}, \citeyear{#1}}
\let\cite\textcite

\DefineBibliographyStrings{english}{
    backrefpage = {\lowercase{p}\adddot },
    backrefpages = {\lowercase{p}p\adddot }
}

\usepackage{tikz}
\usetikzlibrary{patterns,intersections,arrows,decorations.pathreplacing,decorations.pathmorphing,positioning,arrows.meta,calc,decorations.markings,shapes.misc,matrix,shapes,fit,tikzmark}
\usetikzlibrary{calc,shadings}
\usepackage{pgfplots}
\pgfplotsset{compat = 1.16}

\pgfdeclarepatternformonly{north east lines wide}%
   {\pgfqpoint{-1pt}{-1pt}}%
   {\pgfqpoint{10pt}{10pt}}%
   {\pgfqpoint{9pt}{9pt}}%
   {
        \pgfsetlinewidth{0.9pt}
        \pgfpathmoveto{\pgfqpoint{0pt}{0pt}}
        \pgfpathlineto{\pgfqpoint{9.1pt}{9.1pt}}
        \pgfusepath{stroke}
    }

\usepackage[
            margin = 1.18in, 
            top    = 1.18in,
            bottom = 1.18in,
            ]{geometry}
\usepackage[flushmargin]{footmisc}

\usepackage[strict]{changepage}

\makeatletter
\def\th@plain{%
      \thm@notefont{}
      \itshape 
}
\def\th@definition{%
      \thm@notefont{}
      \normalfont 
}
\makeatother

\newtheorem{theorem}{\sc Theorem}
\newtheorem{proposition}{\sc Proposition}

\newtheorem{lemma}{\sc Lemma}
\newtheorem{lemma-app}{\sc Lemma}[section]

\newtheorem{corollary}{\sc Corollary}
\newtheorem{definition}{\sc Definition}
\newtheorem{definition-app}{\sc Definition}[section]

\theoremstyle{definition}

\newtheorem{example}{\sc Example}

\usepackage{thmtools, thm-restate}

\usepackage{xpatch}

\providecommand{\proofnamefont}{\itshape}
\xpatchcmd{\proof}{\itshape}{\normalfont\proofnamefont}{}{}
\renewcommand{\proofnamefont}{\bfseries}

\usepackage{enumitem}

\setlist[enumerate, 1]{
      itemsep   = 0.05em, 
      parsep    = 0.1em, 
      partopsep = 0.1em, 
      topsep    = 0.1em, 
      }

\setlist[itemize, 1]{
      itemsep   = 0.05em, 
      parsep    = 0.1em, 
      partopsep = 0.1em, 
      topsep    = 0.1em, 
      label     = $\bullet$,
      }

\setlist[itemize, 2]{
      itemsep   = 0.15em, 
      parsep    = 0.4em, 
      partopsep = 0.45em,
      label     = --,
      }

\setlist[itemize, 3]{
      itemsep   = 0.1em, 
      parsep    = 0.4em, 
      partopsep = 0.45em,
      label     = $\triangleright$
      }

\def\addlegendimage{\csname pgfplots@addlegendimage\endcsname}

\usepackage{nameref}
\usepackage[capitalise,noabbrev,nameinlink]{cleveref}
\crefname{table}{Table}{tables}
\crefname{equation}{Equation}{equations}
\crefname{lemma-app}{Lemma}{lemmas}
\crefname{claim}{Claim}{claims}
\crefname{problem}{Problem}{problems}
\crefname{example}{Example}{examples}
\crefname{assumption}{Assumption}{Assumptions}
\crefname{definition}{Definition}{Definition}
\crefname{definition-app}{Definition}{Definition}

\newcommand{\comment}[1]{}

\usepackage{algorithm}
\usepackage{algpseudocode}
\algtext*{EndIf}

\DeclareMathOperator*{\argmax}{arg\,max}

\usepackage{comment}

\usepackage[xcolor,leftbars]{changebar}

{\begin{quote}\itshape
  \begin{changebar}\cbcolor{black}\color{black}}%
  {\end{changebar}%
\end{quote}}

\makeatletter
\def\@fnsymbol#1{\ensuremath{%
  \ifcase#1
  \or \star 
  \or $*$        
  \or \dagger
  \or \ddagger
  \or \mathsection
  \or \mathparagraph
  \or \|
  \or \dagger\dagger
  \or \ddagger\ddagger
  \else\@ctrerr
  \fi}}
\makeatother

\begin{document}
\begin{titlepage}
    \title{Persuasion with Verifiable Information\thanks{This paper benefited immensely from the excellent comments and suggestions of the editor, Marzena Rostek, the associate editor, and two anonymous referees. We are indebted to Renee Bowen, Andreas Kleiner, and Joel Sobel for their guidance and support. We also thank Nageeb Ali, Arjada Bardhi, Simone Galperti, Germ\'{a}n Gieczewski, Navin Kartik, Jan Knoepfle, Frédéric Koessler, Remy Levin, Aleksandr Levkun, Elliot Lipnowski, Claudio Mezzetti, Denis Shishkin, Joel Watson, Mark Whitmeyer, as well as numerous seminar and conference audiences, for their helpful feedback. All remaining errors are our own.}}
    \vspace{.2in}
    
    \author{Maria Titova
        \thanks{Corresponding author.}
        \thanks{Department of Economics and Department of Political Science, Vanderbilt University, 2301 Vanderbilt Place, Nashville, TN 37235, USA. Email: \href{mailto:motitova@gmail.com}{\texttt{motitova@gmail.com}}.
        }
        \and Kun Zhang
        \thanks{
            School of Economics, University of Queensland, Level 6, Colin Clark Building 39, Brisbane St. Lucia, QLD 4072, Australia.
            Email: \href{mailto:kun@kunzhang.org}{\texttt{kun@kunzhang.org}}.
        }
    }

    \date{\today}

    \maketitle

    \begin{abstract}
        This paper studies a game in which an informed sender with state-independent preferences uses verifiable messages to convince a receiver to choose an action from a finite set. We characterize the equilibrium outcomes of the game and compare them with commitment outcomes in information design. We provide conditions under which a commitment outcome is an equilibrium outcome and identify environments in which the sender does not benefit from commitment power. Our findings offer insights into the interchangeability of verifiability and commitment in applied settings.
    \end{abstract}



    \thispagestyle{empty}

\end{titlepage}

\section{Introduction}

Persuasion with verifiable information plays an essential role in many economic settings, including courtrooms, electoral campaigns, product advertising, financial disclosure, and job market signaling. In a courtroom, a prosecutor tries to persuade a judge to convict a defendant by selectively presenting inculpatory evidence. In an electoral campaign, a politician carefully chooses which campaign promises he can credibly make to win over voters. In advertising, a firm convinces consumers to purchase its product by highlighting specific product characteristics. In finance, a CEO discloses certain financial statements and indicators to board members to obtain higher compensation. In a labor market, a job candidate lists specific certifications to make her application more attractive to an employer.

We consider the following model of persuasion with verifiable information. First, the sender (he/him) learns the state of the world. Second, the sender chooses a message, which is a verifiable statement about the state of the world, and sends it to the receiver (she/her). Verifiability requires that any feasible message contain the truth (the true state of the world) but not necessarily the whole truth; the message may also include other states. Upon observing the message, the receiver takes an action from a finite set. The sender's preferences are state-independent and strictly increasing in the receiver's action, whereas the receiver's preferences depend on both her action and the state.

Seminal papers in this literature (e.g., \citenop{Grossman1981}, \citenop{Milgrom1981}) establish an ``unraveling'' result, which states that the sender fully reveals the state in every equilibrium. In these papers, the sender's preferences are strictly monotone in the receiver's action (e.g., he is maximizing quantity sold) and the receiver's action space is rich (e.g., she is choosing a perfectly divisible quantity to buy). The argument goes as follows: the sender who is privately informed about the quality of his product always wants to separate himself from all lower-quality senders, as this separation convinces the receiver to purchase a strictly higher quantity of the product. We note that if the receiver's action space is finite, the sender may not fully reveal the state in every equilibrium.
This is easiest to see when the receiver's action space is binary, such as when she is choosing between buying and not buying. In this case high-quality senders may not mind pooling with some lower-quality senders as long as the receiver chooses to buy.

Our first result characterizes (perfect Bayesian) equilibrium outcomes, which we define as mappings from the state space to a distribution over the receiver's actions.
In \cref{thm:outcome_chara}, we show that every equilibrium outcome must be incentive-compatible (for the sender, IC for short) and obedient (for the receiver).
We say that an outcome is IC if the sender receives at least his complete information payoff in each state; otherwise, he has a profitable deviation toward fully revealing the state.
Obedience requires that if the receiver takes an action with positive probability in some states, it must maximize her expected utility.
The second part of \cref{thm:outcome_chara} adds that if an outcome is deterministic, IC, and obedient, then it is an equilibrium outcome.
A deterministic outcome is one in which the receiver takes some action with probability one in every state. Although not all equilibrium outcomes are deterministic, we show in \cref{lem:mixed_eq} that all equilibria in which the receiver does not mix (e.g., equilibria in which the receiver uses a predetermined tie-breaking rule) induce deterministic outcomes.

In our model, the sender does not have commitment power: he learns the state and then chooses a verifiable message that maximizes his expected payoff in that state.
Our second goal is to understand when the sender can achieve the same payoff in equilibrium as he does in information design (e.g., \citenop{Kamenica2011}).
In information design, the sender commits to a disclosure strategy before learning the state; a commitment outcome is an obedient outcome that maximizes the sender's ex-ante utility.
Our second main result (\cref{thm:det_cmt_outcome}) states that the commitment payoff is achievable in equilibrium if and only if there exists a commitment outcome that is deterministic and IC. Intuitively, commitment outcomes that are IC but nondeterministic are generally not equilibrium outcomes because in a commitment outcome, the receiver typically breaks ties in favor of the sender-preferred action. However, as we mentioned earlier, all equilibria in which the receiver does not mix induce deterministic outcomes.

To determine when a commitment outcome can be implemented as an equilibrium, we must ask when a deterministic and IC commitment outcome exists. 
When the state space is rich, we show that a deterministic commitment outcome always exists (\cref{prop:rich_cmt_outcome}). That commitment outcome is an equilibrium outcome if and only if the sender receives at least his complete information payoff by \cref{thm:det_cmt_outcome}. When the state space is finite, however, {\it all}  commitment outcomes may be nondeterministic (e.g., in the seminal example of \citenop{Kamenica2011}). We show that in a modified game in which the set of available verifiable messages is determined stochastically, it is possible to implement any IC commitment outcome (not only a deterministic one; see \cref{section:SMM}).

Throughout the paper, we illustrate our results for the special case in which the receiver chooses between two actions, a setting commonly used in applications.\footnote{See, for example, \cite{Kolotilin2015}, in which pharmaceutical companies persuade the U.S. Food and Drug Administration to approve drugs; \cite{Ostrovsky2010} and \cite{Boleslavsky2015}, in which schools persuade employers to hire their graduates; \cite{Alonso2016} and \cite{Bardhi2018}, in which politicians persuade voters; and \cite{Gehlbach2014}, in which governments persuade citizens.}
For this case we show that an IC commitment outcome always exists (\cref{prop:2actionsIC}).
Thus, when the receiver has two actions, verifiability and commitment assumptions are interchangeable when the state space is sufficiently rich (\cref{prop:binary_rich,prop:binary_finite_approx}).

\subsubsection*{Related Literature} \label{ss:lit_rev}

The literature on verifiable disclosure (games in which the sender learns the state and then chooses a message out of a state-dependent message space) was pioneered by \cite{Grossman1980}, \cite{Grossman1981}, and \cite{Milgrom1981}; this paper uses the same mapping from states to available messages as in \cite{Milgrom1986}, except in \cref{section:SMM}.\footnote{For detailed surveys of this literature, see, for example, \cite{Milgrom2008} and \cite{Dranove2010}.}

A few recent papers similarly characterize the equilibrium set (or the set of equilibrium payoffs of the sender) and assess the value of commitment in various verifiable disclosure models.
\cite{Zhang2022} focuses on a special case of our model, further assuming that the state space is a unit interval, the receiver has monotone preferences, and the receiver's optimal action only depends on the expected state. Under these assumptions, the information design problem is known to have a bi-pooling solution, which always induces a deterministic commitment outcome. \cite{Zhang2022} provides conditions under which this solution is implementable in equilibrium.
\cite{Ali2024} focus on settings in which the sender favors uncertainty: his preferences are state-dependent and deviations to full revelation are never profitable. They provide conditions under which the sets of equilibrium payoffs of the sender are virtually the same in the disclosure game as in information design.
\cite{GieczewskiTitova} consider a generalized disclosure game with an arbitrary message mapping and focus on coalition-proof equilibria.

Outside of verifiable disclosure models, our paper also relates to the informed information design (IID) literature pioneered by \cite{PerezRichet2014}, especially \citeauthor{KoesslerSkreta2023} (\citeyear{KoesslerSkreta2023}; KS henceforth) and \citeauthor{Zapechelnyuk2023} (\citeyear{Zapechelnyuk2023}; Z henceforth).
In IID, the sender chooses a Blackwell experiment like in information design, except he observes the state of the world before making the choice.
Therefore, in IID, a sender faces additional incentive-compatibility constraints relative to (uninformed) information design, much like in verifiable disclosure.
The key difference between IID and disclosure games is that the sender can use stochastic evidence in IID, whereas his evidence in verifiable disclosure is deterministic.
The differences in equilibrium sets between IID and our model highlight the value of stochastic evidence.\footnote{
    Equilibrium concepts differ across all aforementioned papers; for a direct comparison of our results to those of KS and Z, we use perfect Bayesian equilibrium (PBE) with the refinement of the principle of preeminence of tests, which requires that ``every out-of-equilibrium posterior belief must assign probability one to each event that is revealed as certain by the test'' (Z, p. 1061). The principle of preeminence of tests rules out non-IC PBE because the receiver learns the state when the sender sends a fully informative experiment, and this deviation must be unprofitable. Note that PBE without refinements has no predictive power in IID, meaning that every obedient outcome is a PBE outcome (KS, p. 3197).
}
In unconstrained IID (KS), an obedient outcome is an equilibrium outcome if and only if it is IC.\footnote{
    KS focuses on interim optimal (IO) outcomes, which are PBE outcomes with a restriction on R's off-path beliefs: any such belief must assign positive probability only to states in which the sender strictly benefits from the deviation.
}
In IID constrained to nondegenerate experiments (Z), every obedient outcome is an equilibrium outcome.
We show that in Milgrom-Roberts's verifiable disclosure, an obedient outcome is an equilibrium outcome if and only if it is IC and {\it deterministic} (assuming that the receiver uses a pure strategy, as in KS and Z). 
Thus, the sender values stochastic evidence when the state space is finite but not when it is rich. 
An IID problem can also be interpreted as a verifiable disclosure game with random certification (where the randomization between messages is done by a machine, not the sender).\footnote{We thank Fr{\'e}d{\'e}ric Koessler for pointing this out.} We formalize this observation in \cref{section:SMM} by introducing a verifiable disclosure game with a stochastic message mapping and showing that its equilibrium set is the same as in unconstrained IID.
We describe the relationship between our results and those of KS in more detail throughout the paper.

While we study when the sender does not benefit from commitment power, a growing body of literature examines how much the receiver gains from commitment power by comparing equilibrium outcomes with those of optimal mechanisms in sender-receiver games with verifiable information. When the sender's preferences are state-independent, \cite{GlazerRubinstein2004,GlazerRubinstein2006} and \cite{Sher2011} find that the receiver does not need commitment power to reach the optimal mechanism outcome. \cite{HartEtAl2017} and \cite{Ben-Porath2019} provide conditions under which the equilibrium and optimal mechanism outcomes are equivalent.

\cite{Chakraborty2010}, \cite{Lipnowski2020}, and \cite{Lipnowski2020b} study cheap-talk games in which the sender has state-independent preferences; the latter two compare equilibrium outcomes in one-shot cheap-talk games with commitment outcomes. In cheap-talk games, the sender's messages are not verifiable: in every state, the sender has access to the same (sufficiently rich) set of messages. The verifiability requirement faced by our sender significantly impacts the set of equilibrium outcomes.\footnote{Verifiability of his messages may help or hurt the sender depending on the preferences of the players. In fact, every equilibrium of the verifiable information game may be ex-ante better for the sender than every cheap-talk equilibrium and vice versa.} 
\cite{KamenicaLin2024} show that in standard cheap-talk games with finitely many actions and states (where ``standard'' means that the receiver is uninformed and the sender has no actions other than the choice of a message), generically the commitment payoff is achieved in an equilibrium if and only if there exists a deterministic commitment outcome. Our \cref{thm:det_cmt_outcome} provides a similar result for verifiable disclosure games.

\section{Model} \label{s:model}

We study a game of persuasion with verifiable information between a sender (S, he/him) and a receiver (R, she/her). Below we describe the timing of the game along with the assumptions:\footnote{For a topological space $Y$, let $\Delta(Y)$ denote the set of Borel probability measures on $Y$. For $\gamma \in \Delta Y$, let $\supp \gamma$ denote the support of $\gamma$. We say that $\gamma \in \Delta Y$ is degenerate if $\supp \gamma$ is a singleton, and non-degenerate otherwise.}
\begin{enumerate}
    \item S observes the state of the world, $\theta \in \Theta$.

          The state space $\Theta$ is either finite ($\Theta = \{ 1,\ldots, N \}$, $N \geq 2$) or rich ($\Theta$ is a convex and compact subset of $\mathbb{R}^n$). The state of the world is drawn from a common prior $\mu_0 \in \Delta (\Theta)$ with $\supp \mu_0 = \Theta$. If the state space is rich, we assume that the prior is atomless.

    \item S sends message $m \in  M $ to R, where $ M $ is the collection of nonempty Borel subsets of $\Theta$. Since each message is a subset of the state space, we interpret it as a statement about the state of the world. S's messages are \emph{verifiable} in the sense that every message must contain the truth: the set of messages available to S in state $\theta \in \Theta$ is $\{m \in M \sep \theta \in m\}$.\footnote{We borrow from \cite{Milgrom1986} the definition of a verifiable message as a subset of the state space that includes the realized state. This method satisfies normality of evidence (\citenop{BullWatson2007}), which makes it consistent with both major ways of modeling hard evidence in the literature.}
    \item R observes the message (but not the state) and takes an action from a finite set $J :=\{1, \ldots, K\}$ with $K \ge 2$.
    \item The game ends, and payoffs are realized.

          S's payoff $v: J \to \mathbb{R}$ depends only on R's action.
          Without loss, we assume that actions are ordered such that $v$ is increasing in $j \in J$.
          For ease of exposition, we also assume that $v$ is strictly increasing.

          R's preferences are described by a bounded measurable utility function $u: J \times \Theta \to \mathbb{R}$. We define R's {\it complete information action-$j$ set} as $A_j := \{ \theta \in \Theta \sep u(j,\theta) \geq u(j',\theta) \text{ for all } j' \in J \}$ to include all the states of the world in which she prefers to take action $j$ under complete information. 
\end{enumerate}

We consider perfect Bayesian equilibria (henceforth equilibria) of this game. First, S's strategy is a function $\sigma: \Theta \to \Delta_0 M $, where $\Delta_0 M $ is the set of probability measures on $ M $ with a finite support.\footnote{
    That is, we assume that S mixes between finitely many messages. This assumption imposes no restriction when $\Theta$ is finite. When $\Theta$ is rich, it guarantees that $\sigma(\cdot \sep \theta)$ is well-defined, and the restriction does not affect the set of achievable equilibrium payoffs.
}
Second, R's strategy is a function $\tau:  M  \to \Delta J$. Finally, R's belief system $q:  M  \to \Delta \Theta$ describes R's beliefs about the state after any observed message.

\begin{definition}\label{dfn:model-equilibrium}
    A triple $(\sigma, \tau, q)$ is an \emph{equilibrium} if
    \begin{enumerate}[label=(\roman*)]
        \item for all $\theta \in \Theta$, $\sigma (\cdot\sep \theta)$ is supported on $\argmax\limits_{\{m \in  M  \sep \theta \in m\}} \sum\limits_{j \in J} v(j) \, \tau(j \sep m)$;\label{eqm_cond_i}
        \item for all $m \in  M $, $\tau (\cdot \sep m)$ is supported on $\argmax\limits_{j \in J} \label{eqm_cond_ii}\int\limits_{\Theta} u(j,\theta) \, \mathrm{d}q(\theta \sep m)$;
        \item $q$ is obtained from $\mu_0$, given $\sigma$, using Bayes' rule whenever possible;\footnote{That is, $q$ is a regular conditional probability system.} \label{eqm_cond_iii}
        \item for all $m \in  M $, $q (\cdot \sep m) \in \Delta m$. \label{eqm_cond_iv}
    \end{enumerate}
\end{definition}

In words, in equilibrium, \ref{eqm_cond_i} S chooses verifiable messages that maximize his expected utility in every state $\theta \in \Theta$; \ref{eqm_cond_ii} R maximizes her expected utility given her posterior belief; \ref{eqm_cond_iii} R uses Bayes' rule to update her beliefs whenever possible; and \ref{eqm_cond_iv} R's posteriors are consistent with disclosure on and off the path.

To analyze the model, we use the following approach.
Let $\Psi$ be the set of all Borel measurable functions from $\Theta$ to $\Delta J$. We refer to any $\alpha \in \Psi$ as an \emph{outcome}; it specifies, for each state $\theta \in \Theta$, the probability $\alpha(j \sep \theta)$ that R takes action $j \in J$. Given a pair of strategies $(\sigma,\tau)$ of S and R, we let $M_j(\sigma,\tau) := \{ m \in  M  \sep m \in \supp \sigma(\cdot\sep\theta) \text{ for some } \theta\in\Theta \text{ and } \tau(j \sep m) > 0\}$ be the set of messages that convince R to take action $j \in J$, sent with a positive probability in some state $\theta \in \Theta$. We say that $\alpha \in \Psi$ is an \emph{equilibrium outcome} if there exists an equilibrium $(\sigma,\tau,q)$ that {\it induces} it, meaning that $\alpha(j \sep \theta) = \sum\limits_{m \in M_j(\sigma,\tau)} \tau(j \sep m) \sigma(m \sep \theta)$ for all $j \in J$ and $\theta \in \Theta$.

We say that an outcome $\alpha \in \Psi$ is \emph{deterministic} if $\alpha(\cdot \sep \theta)$ is degenerate for each $\theta \in \Theta$. For a deterministic outcome $\alpha$, we refer to the collection of sets $\{ W_j \}_{j \in J}$, where $W_j := \{ \theta \in \Theta \sep \alpha(j \sep \theta) = 1 \}$, as the \emph{outcome partition} (into subsets $W_j$ of the state space in which R takes action $j \in J$ with probability one) of $\alpha$.

Given an outcome $\alpha$, we let $v_{\alpha}(\theta) := \sum\limits_{j \in J} v(j) \, \alpha(j \sep \theta)$ be S's interim (expected) payoff in state $\theta \in \Theta$ and $V_\alpha := \int\limits_\Theta v_{\alpha}(\theta) \, \mathrm{d}\mu_0(\theta)$ be S's ex-ante utility.

\section{Equilibrium Analysis} \label{s:analysis}

We begin by establishing the lower bound on S's payoff in an equilibrium outcome $\alpha$. One thing that S can do in state $\theta$ is fully reveal it by sending message $\{ \theta \}$ with probability one. Upon receiving message $\{ \theta \}$, R learns that the state is $\theta$ and takes an action that is a best response under complete information. Thus, S's equilibrium payoff in state $\theta$ is bounded below by $\underline{v}(\theta):= \min\limits_{j\in J \text{ s.t. }\theta \in A_j} v(j)$. We refer to this condition as S's IC constraint:\footnote{In fact, $\underline{v}(\theta)$ is the lower bound on S's equilibrium payoff in state $\theta$, meaning there exists an equilibrium in which S's interim payoff is exactly $\underline{v}(\theta)$ for each $\theta \in \Theta$. In this equilibrium, S fully reveals every state, R takes the lowest action that is a best response under complete information, and R's beliefs are skeptical off-path (we define R's skeptical beliefs in the proof of \cref{thm:outcome_chara}).}
\begin{equation}\tag{IC$_\theta$}\label{IC-theta}
    v_\alpha(\theta) \geq \underline{v}(\theta). \vspace{-.3in}
\end{equation}
\begin{definition}
    An outcome $\alpha$ is \emph{incentive-compatible} (IC) if it satisfies \eqref{IC-theta} for each state $\theta\in\Theta$.
\end{definition}

Next, in equilibrium, if R finds it optimal to play action $j$ after several messages, that action must remain optimal even if R does not know which of these messages was sent. 
We can thus ``bundle'' all these messages into a single ``recommendation,'' giving rise to R's obedience constraint for action $j$:
\begin{equation}\tag{obedience$_j$}\label{obedience}
    \int\limits_{\Theta} ( u(j,\theta) - u(j',\theta) ) \alpha(j \sep  \theta) \, \mathrm{d}\mu_0 (\theta) \geq 0, \quad \text{for all } j' \in J.
\end{equation}
~\\[-2\baselineskip]

\begin{definition}
    An outcome $\alpha$ is \emph{obedient} if it satisfies \eqref{obedience} for each action $j \in J$.
\end{definition}

If $\alpha$ is a deterministic outcome with partition $\{ W_j \}_{j\in J}$, then \eqref{IC-theta} becomes $\theta \in W_j \implies v(j) \geq \underline{v}(\theta) \iff j \geq \min\limits_{i \in J \text{ s.t. } \theta \in A_i} i$, indicating that the action taken in state $\theta$ must be no lower than R's worst best response under complete information. The obedience constraint for action $j$ simplifies to $\int\limits_{W_j} ( u(j,\theta) - u(j',\theta) ) \, \mathrm{d}\mu_0 (\theta) \geq 0$ for all $j' \in J$.

Our first result confirms that every equilibrium outcome is IC and obedient. For deterministic outcomes, these two properties are necessary and sufficient for equilibrium implementation.

\begin{theorem} \label{thm:outcome_chara}~
    \begin{enumerate}[label=(\alph*)]
        \item Every equilibrium outcome is IC and obedient. \label{thm1a}
        \item If a deterministic outcome is IC and obedient, then it is an equilibrium outcome. \label{thm1b}
    \end{enumerate}
\end{theorem}
\begin{proof}
    [Part \ref{thm1a}] Consider an equilibrium $(\sigma,\tau,q)$ with outcome $\alpha \in \Psi$. Observe that $\alpha$ must be IC, or else there exists a state $\theta$ in which S has a profitable deviation to fully revealing the state. Next, we show that $\alpha$ is also obedient. Consider any action $j \in J$. By the equilibrium condition \ref{eqm_cond_ii}, we have
    \begin{align*}
        \text{ for all } m \in  M_j(\sigma,\tau) \text{ and } j' \in J, \quad
         & \int\limits_\Theta \left(u(j,\theta) - u(j',\theta)\right) \,\mathrm{d}q(\theta \sep m) \geq 0         \\
        \implies
         & \int\limits_\Theta ( u(j,\theta) - u(j',\theta) ) \tau(j \sep m) \,\mathrm{d}q(\theta \sep m) \geq  0,
    \end{align*}
    where the second inequality follows because $\tau(j \sep m) > 0$ for all $m \in  M_j(\sigma,\tau)$. Using Bayes' rule, the above inequality implies that
    \begin{align*}
        \text{for all } j' \in J, \quad
         & \int\limits_\Theta ( u(j,\theta) - u(j',\theta) ) \sum\limits_{m \, \in \, M_j(\sigma,\tau)} \tau(j \sep m) \sigma(m \sep \theta) \,\mathrm{d}\mu_0(\theta) \geq 0, \\
        \implies
         & \int\limits_\Theta ( u(j,\theta) - u(j',\theta) )\alpha(j \sep \theta) \,\mathrm{d}\mu_0(\theta) \geq 0,
    \end{align*}
    where the last inequality is \eqref{obedience}. Since $j$ was chosen arbitrarily, $\alpha$ is obedient.

        [Part \ref{thm1b}] Consider a deterministic outcome $\alpha$ that is IC and obedient, and denote its outcome partition by $\{ W_j \}_{j \in J}$.
    We construct an equilibrium $(\sigma,\tau,q)$ that induces $\alpha$.
    Let S's strategy be $\sigma(m \sep \theta) = \mathbbm{1}(m = W_j \text{ and }\theta \in W_j)$, which reveals the element of the outcome partition that the realized state belongs to.
    When R receives an on-path message $W_j$, she learns that $\theta \in W_j$ and nothing else; by \eqref{obedience}, playing action $j$ is a best response; thus, we let $\tau(j \sep W_j) = 1$ for all $j \in J$.
    For off-path messages, assume R is ``skeptical'' and believes that any unexpected message comes from the state in which R prefers to take the lowest action under complete information.
    Formally, for all $m \notin \{ W_j \}_{j \in J}$, let $q(\cdot \sep m) \in \Delta(m \cap A_{\underline{j}})$, where $\underline{j} \in J$ is the lowest action $i \in {J}$ such that the set $m \cap A_{i}$ is nonempty. Then playing action $\underline{j}$ with probability one is a best response to message $m$, so we let $\tau(\underline{j}\sep m) = 1$.

    We now show that S has no profitable deviations using the fact that $\{ W_j \}_{j\in J}$ is a partition of the state space.
    Consider a state $\theta \in \Theta$, which is in $W_j$ for some action $j \in J$.
    S cannot send any other on-path message because $\theta \in W_j$ implies $\theta \notin W_i$ for any $i \neq j$. Therefore $W_i$ is not a verifiable message in state $\theta$.
    If S deviates to an off-path (verifiable) message $m \notin \{ W_j \}_{j \in J}$, then S's payoff is $v(\underline{j}) \leq \underline{v}(\theta)$, and this deviation is unprofitable by \eqref{IC-theta}. Therefore, $(\sigma,\tau,q)$ is an equilibrium that induces $\alpha$.
\end{proof}

Part \ref{thm1b} of \cref{thm:outcome_chara} characterizes the set of deterministic equilibrium outcomes, and its proof suggests a simple way of implementing these outcomes in a pure-strategy equilibrium with at most $K$ on-path messages that essentially serve as action recommendations. Specifically, if $\{ W_j \}_{j \in J}$ is an outcome partition, then $W_j$ serves as both the set of states in which R plays action $j$ and the on-path message recommending action $j$ in the constructed equilibrium inducing this outcome.

While \cref{thm:outcome_chara} fully characterizes the set of deterministic equilibrium outcomes, it does not provide a full characterization of the entire set of equilibrium outcomes. In general, an IC and obedient nondeterministic outcome may or may not be an equilibrium outcome. Consider the seminal example from \cite{Kamenica2011}.

\begin{example} \label{example:KG11}
    Suppose S is a prosecutor and R is a judge. The state of the world is binary: $\Theta = \{ 1,2 \} = \{\text{innocent}, \text{guilty}\}$;
    R's action space is binary: $J = \{ 1,2 \} = \{\text{acquit}, \text{convict}\}$; and the prior is $\mu_0(1) = 0.7$. S's preferences are $v(1) = 0$ and $v(2) = 1$, and R's objective is to ``match the state'': $u(1, 1) = u(2, 2) = 1$, and $u(1, 2) = u(2, 1) = 0$.  Consider an outcome $\alpha^*$ in which $\alpha^*(2 \sep 2) = 1$ and $\alpha^*(2 \sep 1) = 3/7$. It is easy to verify that $\alpha^*$ is both IC and obedient. However, $\alpha^*$ is not an equilibrium outcome: when $\theta = 1$, R convicts with probability $3/7$ and acquits with probability $4/7$. Since S strictly prefers conviction, he has a profitable deviation to sending the message after which R convicts when $\theta = 1$.
\end{example}

\cref{example:KG11} illustrates that (IC and obedient) outcomes in which S receives different payoffs from different messages in the same state cannot be equilibrium outcomes.
Once the state is realized, S's message space becomes fixed and known. 
Thus, if S mixes between multiple messages in the same state, he must receive the same payoff from each of these messages. 
Of course, if S does receive the same payoff in every state, then an IC, obedient, and nondeterministic outcome could be an equilibrium outcome. However, in any such equilibrium, R must play a mixed strategy:
\begin{lemma}\label{lem:mixed_eq}
    Suppose that $\alpha$ is a nondeterministic outcome induced by an equilibrium $(\sigma,\tau,q)$. Then in each state $\theta \in \Theta$ such that $\alpha(\cdot \sep \theta)$ is nondegenerate, R is playing a mixed strategy (meaning $\tau(\cdot \sep m)$ is nondegenerate) for some $m\in \supp \sigma(\cdot \sep \theta)$.
\end{lemma}
\begin{proof}
    Let $\theta \in \Theta$ be a state such that $\alpha(\cdot \sep \theta)$ is nondegenerate. By contradiction, suppose that $\tau(\cdot \sep m)$ is degenerate for all $m \in \supp \sigma(\cdot \sep \theta)$.
    By equilibrium condition \ref{eqm_cond_i}, for any pair of messages $m,m' \in \supp \sigma(\cdot \sep \theta)$, we have $\sum\limits_{j \in J} v(j) \tau (j \sep m) = \sum\limits_{j \in J} v(j) \tau (j \sep m')$, implying that there exists an action $j^* \in J$ such that $\tau(j^* \sep m) = \tau(j^* \sep m') = 1$.
    In other words, if R is not mixing, every message sent by S in state $\theta$ leads R to take the same action.
    Therefore, $\alpha(j^* \sep \theta) = \sum\limits_{m \in M_j(\sigma,\tau)} \tau(j^* \sep m) \sigma(m \sep \theta) = 1$, which is a contradiction.
\end{proof}

The contrapositive of \cref{lem:mixed_eq} also tells us that if R is not mixing in an equilibrium (e.g., if she uses an exogenously given tie-breaking rule like in IID), then an obedient outcome is an equilibrium outcome {\it if and only if} it is IC and deterministic.
\cref{thm:outcome_chara} and \cref{lem:mixed_eq} together highlight the difference in equilibrium sets between our verifiable disclosure game and IID (KS and Z).
KS's characterization (Proposition 2) states that an outcome is interim optimal (IO) if and only if it is obedient and IOC, where IOC essentially requires that for every set of states $Q$, and for every state in $Q$, S does not strictly prefer R having a belief supported on $Q$.\footnote{In contrast, IC only requires that in any given state, S does not strictly prefer inducing the degenerate belief at that state. Interestingly, KS also show that IOC and IC are equivalent if S’s value function is quasiconvex in R's belief (Proposition 3) or if R chooses between two actions (Lemma B.2).}
Naturally, the first difference---IOC in KS's setting versus IC in ours---arises from the difference in equilibrium selection, as they impose a stronger restriction on off-path beliefs than we do.
The second difference is that IOC and obedience are necessary and sufficient for an outcome to be IO, whereas for us IC and obedience are not sufficient.
Since in our model S chooses messages, an additional restriction applies: S can mix between different messages only if each message yields the same expected payoff---a constraint absent in IID.
For this reason, some nondeterministic IO, and thus IC, outcomes are not equilibrium outcomes in our game (e.g., one from \cref{example:KG11}).

\cref{thm:outcome_chara} characterizes all pure-strategy equilibria of the game, as these equilibria are deterministic. \cite{KoesslerRenault2012} find that IC and obedience are necessary and sufficient for a pure-strategy outcome to be an equilibrium outcome in a setting in which S has state-independent preferences, sends verifiable messages, and sets a price and R chooses between two actions. \cref{thm:outcome_chara} highlights that this result (1) extends to cases in which R has more than two actions and (2) is not driven by S's additional choice variable (price).\footnote{However, as the authors point out, the price choice in their setting ensures that R plays a pure strategy in equilibrium.}

\section{Value of Commitment}\label{section:4}

In this section, we ask when a commitment outcome, a solution to the information design problem, is also an equilibrium outcome.
In the information design problem, Stage 1 of the game (in which S learns the state) is removed, and Stage 2 of the game (in which S chooses a verifiable message) is replaced by S committing to an experiment that sends signals depending on state realizations.\footnote{
    An experiment $(\mathcal{S}, \chi)$ consists of a compact metrizable space $\mathcal{S}$ of signals and a Borel measurable function $\chi: \Theta \to \Delta \mathcal{S}$.
    R observes the choice of the experiment and a signal realization $s \in \mathcal{S}$ drawn from $\chi(\cdot \sep \theta)$, where $\theta$ is the realized state.
} Importantly, when S has commitment power, he no longer faces incentive-compatibility constraints, i.e., he does not need to maximize his utility state by state.

Following \cite{Kamenica2011}, we focus on straightforward signals that R interprets as action recommendations. Therefore, an (optimal) {\it commitment outcome} $\overline{\psi} \in \Psi$ solves
\begin{equation}\label{eqn:SPcommitment_outcome}\tag{CO}
    \begin{aligned}
        \max\limits_{\psi \in \Psi} \  V_\psi \quad
         & \text{ subject to, for each action } j \in J,                                                      \\
         & \int\limits_{\Theta} ( u(j,\theta) - u(j',\theta) ) \psi(j \sep \theta) \,\mathrm{d}\mu_0 (\theta)
        \geq 0
        \quad \text{for all } j' \in J.
    \end{aligned}
\end{equation}
Simply put, a commitment outcome is an obedient outcome that maximizes S's ex-ante utility. We refer to the value of problem \eqref{eqn:SPcommitment_outcome} as the \emph{commitment payoff}. Our second result shows that a commitment outcome must be deterministic and IC to be an equilibrium outcome.

\begin{theorem} \label{thm:det_cmt_outcome}
    Consider a commitment outcome $\overline{\psi} \in \Psi$.
    \begin{enumerate}[label=(\alph*)]
        \item If $\overline{\psi}$ is IC and deterministic, then it is an equilibrium outcome.\label{thm2:a}
        \item If $\overline{\psi}$ is an equilibrium outcome, then it is IC and $\mu_0$-almost everywhere deterministic.\label{thm2:b}
    \end{enumerate}
\end{theorem}

\begin{proof}

    [Part \ref{thm2:a}] Recall that every commitment outcome is obedient. Therefore, if $\overline{\psi}$ is IC and deterministic, Part \ref{thm1b} of \cref{thm:outcome_chara} implies that it is an equilibrium outcome.

        [Part \ref{thm2:b}] Suppose a commitment outcome $\overline{\psi}$ is an equilibrium outcome, meaning that there exists an equilibrium $(\sigma,\tau,q)$ that induces it. By \cref{thm:outcome_chara}, $\overline{\psi}$ is IC. We will now show that $\overline{\psi}$ is deterministic $\mu_0$-almost everywhere. Define $T := \{ \theta \in \Theta \sep \overline{\psi}(\cdot \sep \theta) \text{ is nondegenerate} \}$ as the set of states in which R plays multiple actions, and suppose, by contradiction, that $\mu_0(T) > 0$. By \cref{lem:mixed_eq}, for each $\theta \in T$, there exists a message $m \in \supp \sigma(\cdot \sep \theta)$ such that $\tau(\cdot \sep m)$ is nondegenerate.
    Let $\widetilde{M} := \{ m \in  M  \sep \tau(\cdot \sep m) \text{ is nondegenerate} \}$ be the set of messages after which R plays a mixed strategy.
    Define $\widetilde{\tau}(j^* \sep m) := \mathbbm{1}( j^* = \max\limits_{j \in \supp \tau(\cdot \sep m) } j )$ for all $m \in M$ as R's strategy that breaks all ties in $\tau$ in favor of S. Denote the outcome from the strategy profile $(\sigma,\widetilde{\tau})$ by $\widetilde{\psi}$.

    We derive a contradiction by showing that $\widetilde{\psi}$ is an obedient outcome with $V_{\widetilde{\psi}} > V_{\overline{\psi}}$, which implies that $\overline{\psi}$ is not a commitment outcome.
    Indeed, we have $v_{\widetilde{\psi}}(\theta) > v_{\overline{\psi}} (\theta)$ for all $\theta \in T$ (since there exists an $m \in \widetilde{M}$ with $\sigma(m \sep \theta) > 0$), while $v_{\widetilde{\psi}}(\theta) = v_{\overline{\psi}} (\theta)$ for all $\theta \notin T$.
    Therefore, $V_{\widetilde{\psi}} - V_{\overline{\psi}} = \int\limits_T ( v_{\widetilde{\psi}}(\theta) - v_{\overline{\psi}}(\theta) ) \,\mathrm{d}\mu_0(\theta) > 0$ since $\mu_0(T) > 0$. To prove that $\widetilde{\psi}$ is obedient, we apply equilibrium condition (ii) to the equilibrium $(\sigma, \tau, q)$ and follow the steps in the proof of  \cref{thm:outcome_chara}\ref{thm1a}, replacing $\tau$ with $\widetilde{\tau}$ and noting that $M_j(\sigma,\widetilde{\tau}) \subseteq  M_j(\sigma,\tau)$.
\end{proof}

The nontrivial part of \cref{thm:det_cmt_outcome} involves proving that if $\overline{\psi}$ is both an equilibrium outcome and a commitment outcome, then it is deterministic almost everywhere. This is equivalent to showing that a nondeterministic equilibrium outcome cannot be a commitment outcome. Indeed, by \cref{lem:mixed_eq}, in the equilibrium that induces $\overline{\psi}$, R must play a mixed strategy following some on-path messages from a positive measure of states.
However, breaking those ties in favor of the S-preferred action strictly increases S's ex-ante utility, which implies that $\overline{\psi}$ is not a commitment outcome.

In many relevant settings, R chooses between two actions.
In this case, the analysis vastly simplifies.
From R's perspective, there are ``bad'' states $\theta \in A_1$, in which R prefers the low action 1, and ``good'' states $\theta \notin A_1$, in which she prefers the high action $2$.
The highest payoff that S can achieve is $v(2)$ (when R takes action $2$ with probability one), and the lowest is $v(1)$.
To state that an outcome $\psi \in \Psi$ is IC, it suffices to show that $\theta \notin A_1$ implies that $v_\psi(\theta) = v(2) \psi(2\sep\theta) + v(1) \psi(1\sep \theta) \geq v(2)$, which is equivalent to $\psi(2 \sep \theta) = 1$. The IC condition for $\theta \in A_1$ is not relevant because $v(1)$ is already the lowest payoff in the game.
In words, an outcome is IC if and only if R plays action 2 with probability one in all states in which action 2 is the unique best response under complete information.
The following result establishes the existence of an IC commitment outcome when R chooses between two actions.

\begin{proposition}\label{prop:2actionsIC}
    If $|J| = 2$, then there exists an IC commitment outcome.
\end{proposition}
\begin{proof}
    Since $\Theta$ is a compact subset of $\mathbb{R}^n$, a commitment outcome exists by Proposition 3 in the online appendix of \cite{Kamenica2011} and Theorem 1 in \cite{Terstiege2023}. Let $\overline{\psi} \in \Psi$ be a commitment outcome and let $\widetilde{\psi} \in \Psi$ be an outcome such that $\widetilde{\psi}(\cdot \sep \theta) = \overline{\psi}(\cdot \sep \theta)$ for all $\theta \notin A_2$ and $\widetilde{\psi}(2\sep \theta) = 1$ for all $\theta \in A_2$. By construction, $\widetilde{\psi}$ is IC and weakly increases S's ex-ante utility over $\overline{\psi}$. Define $\delta(\theta) := u(2,\theta)-u(1,\theta)$ and observe that
    \begin{align*}
        \int\limits_\Theta \delta(\theta) \widetilde{\psi} (2\sep \theta) \,\mathrm{d}\mu_0(\theta)
        = \int\limits_\Theta \delta(\theta) \overline{\psi} (2\sep \theta) \,\mathrm{d}\mu_0(\theta)
        + \int\limits_{A_2} \delta(\theta) (1-\overline{\psi}(2\sep \theta)) \,\mathrm{d}\mu_0(\theta),
    \end{align*}
    where the last term is nonnegative because $\delta(\theta) \geq 0$ for all $\theta \in A_2$. Consequently, obedience of $\overline{\psi}$ (for both actions) implies obedience of $\widetilde{\psi}$. Hence, $\widetilde{\psi}$ is also a commitment outcome.
\end{proof}

The existing literature provides additional insights into commitment outcomes when $|J| = 2$ and $\Theta$ is finite.
\cite{Alonso2016} show that every commitment outcome is characterized by a cutoff state $\theta^*$, with all states satisfying $\delta(\theta) > \delta(\theta^*)$ pooled together to recommend action 2.
In particular, in all good states $\theta \notin A_1$, S recommends action 2, which implies that every commitment outcome is IC (see also Lemma B.2 in \citenop{KoesslerSkreta2023}).
Our \cref{prop:2actionsIC} also addresses the case in which $\Theta$ is rich. In this case, some commitment outcomes are not IC (although they are IC $\mu_0$-almost everywhere), and its proof outlines how to make an existing commitment outcome incentive-compatible.

Returning to the more general case in which $J \geq 2$, \cref{thm:det_cmt_outcome} is useful for verifying whether an existing commitment outcome $\overline{\psi}$ is an equilibrium outcome.
The answer is affirmative if and only if $\overline{\psi}$ is deterministic $\mu_0$-a.e. and IC.
Although verifying incentive compatibility may be straightforward, a deterministic commitment outcome is not guaranteed to exist.
In the remainder of this section, we consider the cases in which $\Theta$ is rich and $\Theta$ is finite separately.
We show that when $\Theta$ is rich, a deterministic commitment outcome always exists.
Furthermore, if $|J|=2$, the commitment payoff is always attained in equilibrium.
When $\Theta$ is finite, we derive an approximation result.

\subsection{Rich State Space}

When the state space $\Theta$ is rich (a convex and compact subset of $\mathbb{R}^n$) and the prior $\mu_0$ is atomless, the existence of a deterministic commitment outcome is guaranteed.

\begin{proposition} \label{prop:rich_cmt_outcome}
    If $\Theta$ is rich, then a deterministic commitment outcome exists. Furthermore, a deterministic commitment outcome is an equilibrium outcome if and only if it is IC.
\end{proposition}

\begin{proof}
    The existence of a commitment outcome $\overline{\psi}$ follows from the same argument as that used in the proof of \cref{prop:2actionsIC}. Furthermore, $\overline{\psi}(j \sep \cdot): \Theta \to [0,1]$ is Borel measurable for every $j \in J$ and $\sum\limits_{j \in J} \overline{\psi}(j \sep \theta) = 1$ for all $\theta \in \Theta$. Let $\mu_j$ be such that $\mathrm{d}\mu_j := u(j, \cdot) \, \mathrm{d} \mu_0$ for each $j \in J$.

    Since $\mu_0$ is a finite and atomless positive measure and $u$ is bounded, $\mu_j$ is a finite and atomless signed measure for each $j \in J$.
    By Theorem 2.1 in \cite{DvoretzkyEtAl1951}, since $J$ is finite, there exist Borel measurable functions $\widetilde{\psi}(j \sep \cdot): \Theta \to \{0,1\}$ for all $j \in J$,
    with $\sum\limits_{j \in J} \widetilde{\psi}(j \sep \cdot) = 1$, such that (I) $\int\limits_{\Theta} \widetilde{\psi}(j \sep \theta) \, \mathrm{d} \mu_0 = \int\limits_{\Theta} \overline{\psi}(j \sep \theta) \, \mathrm{d} \mu_0$ and (II) $\int\limits_{\Theta} \widetilde{\psi}(j \sep \theta) \, \mathrm{d} \mu_j = \int\limits_{\Theta} \overline{\psi}(j \sep \theta) \, \mathrm{d} \mu_j$ for all $j \in J$.
    Condition (I) implies that
    \[
        V_{\widetilde{\psi}} = \int\limits_{\Theta} \sum_{j \in J} v(j) \widetilde{\psi}(j \sep \theta) \, \mathrm{d}\mu_0 = \int\limits_{\Theta} \sum_{j \in J} v(j) \overline{\psi}(j \sep \theta) \, \mathrm{d}\mu_0 = V_{\overline{\psi}}.
    \]
    Condition (II) implies that $\widetilde{\psi}$ is obedient, as $\overline{\psi}$ is.
    Hence, $\widetilde{\psi}$ is a deterministic commitment outcome. The second part follows from \cref{thm:det_cmt_outcome}.
\end{proof}

Verifying whether a deterministic commitment outcome with partition $\{ W_j \}_{j \in J}$ is IC (and therefore an equilibrium outcome) is straightforward. It requires determining whether $\theta \in W_j$ implies $v(j) \geq \underline{v}(\theta)$ for all $\theta \in \Theta$.
Consider the following example from \cite{Gentzkow2016}.

\begin{example}\label{example:GK16}
    Suppose R has three actions, $J = \{ 1,2,3 \}$, and the prior is uniform on $\Theta = [0,1]$. S's payoffs are given by $v(1) = 0$, $v(2) = 1$, and $v(3) = 3$. R's preferences depend only on the posterior mean. Given belief $\mu \in \Delta \Theta$, action $1$ is optimal if and only if $\mathbbm{E}_\mu[\theta] \le 1/3$; action 2 is optimal if and only if $\mathbbm{E}_\mu[\theta] \in [ 1/3, 2/3]$; and action 3 is optimal if and only if  $\mathbbm{E}_\mu[\theta] \geq 2/3$. Therefore, R's complete-information action sets are $A_1 = [ 0,1/3 ]$, $A_2 = [ 1/3, 2/3]$, and $A_3 = [2/3,1]$. \cite{Gentzkow2016} identify a deterministic commitment outcome $\overline{\psi}$ with an outcome partition $\overline{W}_1 = [0,8/48)$, $\overline{W}_2 = (11/48,21/48)$, and $\overline{W}_3 = [8/48,11/48] \cup [21/48,1]$.
    This outcome is IC, which we illustrate in \cref{fig:GKexample3}.
    Since $\overline{\psi}$ is a deterministic and IC commitment outcome, it is an equilibrium outcome by \cref{prop:rich_cmt_outcome}.

    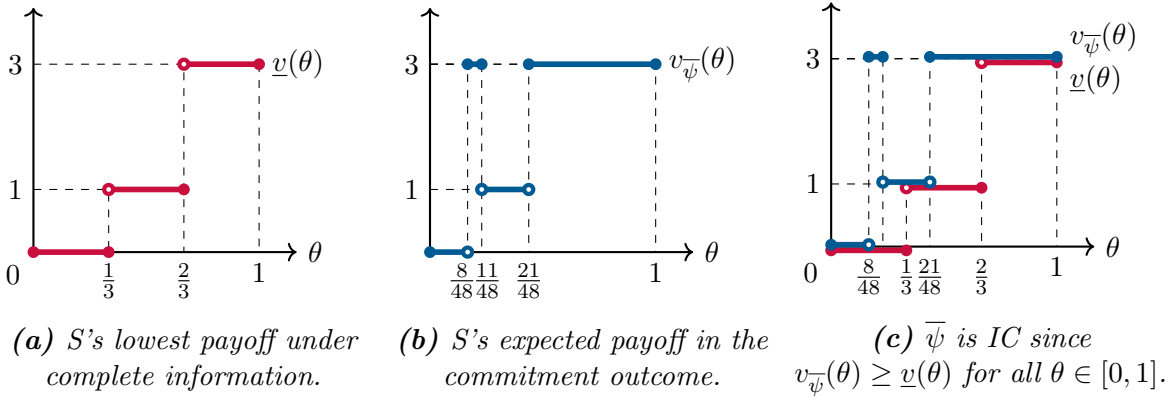
\begin{figure}[ht!]
        \centering
        \begin{subfigure}[b]{0.32\textwidth}
            \input{fig-v-underline}
            \caption{S's lowest payoff under complete information.}
        \end{subfigure}
        \hfill
        \begin{subfigure}[b]{0.32\textwidth}
            \input{fig-v-psi}
            \caption{S's expected payoff in the commitment outcome.}
        \end{subfigure}
        \hfill
        \begin{subfigure}[b]{0.32\textwidth}
            \input{fig-both}
            \caption{$\overline{\psi}$ is IC since $v_{ \overline{\psi} }(\theta) \geq \underline{v}(\theta)$ for all $\theta \in [0,1]$.}
        \end{subfigure}
        \caption{Commitment outcome $\overline{\psi}$ is IC since S receives at least his complete-information payoff in every state of the world.}
        \label{fig:GKexample3}
    \end{figure}

\end{example}

When R chooses between two actions, S always attains his commitment payoff in equilibrium.

\begin{proposition} \label{prop:binary_rich}
    If $\Theta$ is rich and $|J| = 2$, then there exists a commitment outcome that is an equilibrium outcome.
\end{proposition}

\begin{proof}
    By \cref{prop:rich_cmt_outcome}, there exists a deterministic commitment outcome $\overline{\psi}$. Using the same argument as that used in the proof of \cref{prop:2actionsIC}, we construct a deterministic commitment outcome $\widetilde{\psi}$ that is IC. By \cref{prop:rich_cmt_outcome}, $\widetilde{\psi}$ is an equilibrium outcome.
\end{proof}

\subsection{Finite State Space}

When the state space is finite, i.e., $\Theta = \{1, \ldots, N\}$, a deterministic commitment outcome may not exist.
For instance, in \cref{example:KG11}, the \emph{unique} commitment outcome is not deterministic.
As a result, S may not be able to achieve the commitment payoff in equilibrium.

However, here we show that when the state space is sufficiently rich (in the sense that $\mu_0(\theta)$ is sufficiently small for each $\theta \in \Theta$), then S's equilibrium payoff approaches his commitment payoff.
For a concise argument, we adopt the assumptions of \cite{Alonso2016}: R has a binary action and
\begin{equation} \label{eqn:alonso}
    \theta' \ne \theta'' \Longrightarrow \delta(\theta') \ne \delta(\theta''), \tag{RU}
\end{equation}
where $\delta(\theta)=u(2,\theta) - u(1, \theta)$ for all $\theta \in \Theta$.

\begin{proposition} \label{prop:binary_finite_approx}
    Suppose that $\Theta$ is finite, $|J| = 2$, \eqref{eqn:alonso} holds,
    and S's payoffs are normalized to $v(2) = 1$ and $v(1) = 0$.\footnote{Normalizing S's payoffs is without loss of generality; condition \eqref{eqn:alonso} simplifies the proof, but the result remains true without it.}
    Let $V^*$ be S's commitment payoff. For every $\varepsilon>0$, there is $\gamma>0$ such that if $\mu_{0}(\theta)<\gamma$ for all $\theta \in \Theta$, then there exists an equilibrium outcome $\alpha$ with $| V^* - V_\alpha | < \varepsilon$.
\end{proposition}

\begin{proof}
    If $A_2 = \Theta$, then let $\alpha(2 \sep \theta) = 1$ for all $\theta \in \Theta$ so that $V_\alpha = V^*$.
    Thus, we assume for the remainder of the proof that $A_2$ is a proper subset of $\Theta$.
    Since \eqref{eqn:alonso} holds, we can use Proposition 2 in \cite{Alonso2016} to find a cutoff state $\theta^* \in \Theta$ such that $\delta(\theta^*) < 0$ and, for every commitment outcome $\psi$, we have $\psi(2\sep \theta) = 1$ ($\psi(1\sep \theta) = 1$) for all $\theta \in \Theta$ such that $\delta(\theta) > \delta(\theta^*)$ ($\delta(\theta) < \delta(\theta^*)$).
    Now consider a deterministic outcome $\alpha$ with partition $\{ W_1,W_2 \}$ such that $W_2 = \{ \theta \in \Theta \sep \psi (2 \sep \theta) = 1\}$ and $W_1 = \Theta \smallsetminus W_2$.
    It is easy to see that $\alpha$ is IC and obedient, and therefore it is an equilibrium outcome by \cref{thm:outcome_chara}.
    If $\psi(2 \sep \theta^*) = 1$, the difference in S's ex-ante payoffs is zero; otherwise, we have $V^* - V_\alpha =  \psi(2 \sep \theta^*) \mu_0(\theta^*)  < \mu_0(\theta^*) < \gamma := \varepsilon$.
\end{proof}

Thus, when R chooses between two actions, S can attain a payoff arbitrarily close to his commitment payoff in equilibrium as long as the prior probability of each state is sufficiently small.

\section{A Model with a Stochastic Message Mapping}\label{section:SMM}

In the main model, IC and obedience are not sufficient for an outcome to be an equilibrium outcome; there exist nondeterministic but IC and obedient outcomes in which S effectively recommends multiple actions, leading to different expected payoffs in the same state. This violates equilibrium condition \ref{eqm_cond_i}. The reason why \ref{eqm_cond_i} is violated is that the mapping $E: \Theta \rightrightarrows M$, a correspondence that determines the set of messages available in state $\theta$, is deterministic.
This assumption is standard in the literature on verifiable disclosure and cheap talk.\footnote{
    For example, in \cite{Grossman1981}, \cite{Milgrom1981}, and \cite{Milgrom1986}, $E(\theta)$ includes subsets of $\Theta$ that contain $\theta$. In \cite{Dye1985}, $E(\theta)$ is binary; S can reveal $\theta$ or say nothing. In \cite{HartEtAl2017} and \cite{Ben-Porath2019}, $E(\theta)$ is a partial order on $\Theta$. In cheap talk, $E(\theta)$ is the same for all $\theta \in \Theta$.
}
In some cases, however, it is reasonable to assume that the mapping $E(\theta)$ is stochastic: for example, there may be different labels for the same state, and S can make statements about the label rather than the state.
In this section, we introduce a verifiable disclosure game with a stochastic message mapping (henceforth, the SMM game) and show that IC and obedience are sufficient for an outcome to be an equilibrium outcome of this game.

The SMM game has the same timeline and player objectives as our main model, with the only modification occurring in Stage 2, in which S communicates with R.
Specifically, we assume that along with the state of the world $\theta \in \Theta$, where the state space $\Theta = \{1, \ldots, N\}$ is finite, S also observes a \emph{label} $x \in [0,1]$, which is payoff-irrelevant to both S and R.
The label $x$ is drawn from the uniform distribution on $X^\theta$,
where $\{ X^\theta \}_{\theta \in \Theta}$ forms a partition of the unit interval such that $\lambda(X^\theta) = \mu_0(\theta)$, where $\lambda$ is the Lebesgue measure.\footnote{
    For example, let $t_0 := 0$, $t_\theta := \sum\limits_{\theta'=1}^\theta \mu_0(\theta')$ for all $  \theta \in \Theta$; also, let $X^\theta = [t_{\theta-1},t_{\theta})$ for all $\theta \in \{ 1,\ldots,N-1 \}$ and $X^N = [t_{N-1},1]$. Then $\{ X^{\theta} \}_{\theta \in \Theta}$ is a partition of $[0,1]$ and $\lambda(X^\theta) = \sum\limits_{\theta'=1}^\theta \mu_0(\theta') - \sum\limits_{\theta'=1}^{\theta-1} \mu_0(\theta') = \mu_0(\theta)$.
} 
Having observed $\theta$ and $x$, S sends message $m \in \widehat{M}$ such that $x \in m$, where $\widehat{M}$ is the collection of nonempty Borel subsets of $[0,1]$. Thus, the set of messages available to S in state $\theta$ is now determined stochastically (through $x$). The {\it equilibrium of the SMM game} $(\widehat{\sigma},\widehat{\tau},\widehat{q})$ is defined analogously to that of the main model, except S's strategy also depends on $x$.

\begin{definition}\label{dfn:SMM-equilibrium}
    A triple $(\widehat{\sigma},\widehat{\tau},\widehat{q})$, where $\widehat{\sigma}: \Theta \times [0,1] \to \Delta_0  \widehat{M}$ is S's strategy, $\widehat{\tau}:  \widehat{M} \to \Delta J$ is R's strategy, and $\widehat{q}:  \widehat{M} \to \Delta \Theta$ is R's belief system, is an \emph{equilibrium of the SMM game} if
    \begin{enumerate}[label=(\roman*)]
        \item for all $\theta \in \Theta$ and $x \in [0,1]$, $\widehat{\sigma} (\cdot\sep \theta,x)$ is supported on $\argmax\limits_{\{m \in  \widehat{M} \sep x \in m\}} \sum\limits_{j \in J} v(j) \, \widehat{\tau}(j \sep m)$;\label{eqm_cond_i-SMM}
        \item for all $m \in  \widehat{M}$, $\widehat{\tau} (\cdot \sep m)$ is supported on $\argmax\limits_{j \in J} \label{eqm_cond_ii-SMM}\int\limits_{\Theta} u(j,\theta) \, \mathrm{d}q(\theta \sep m)$;
        \item $\widehat{q}$ is obtained from $\mu_0$, given $\widehat{\sigma}$, using Bayes' rule; \label{eqm_cond_iii-SMM}
        \item for all $m \in  \widehat{M}$, $\widehat{q} (\cdot \sep m) \in \Delta  \{ \theta \in \Theta \sep X^\theta \cap m \neq \varnothing \}$. \label{eqm_cond_iv-SMM}
    \end{enumerate}
\end{definition}

Since $x$ is payoff-irrelevant, an outcome $\alpha$ of the SMM game is an element of $\Psi$. An outcome $\alpha$ is an \emph{equilibrium outcome} of the SMM game if an equilibrium $(\widehat{\sigma},\widehat{\tau},\widehat{q})$ exists that induces it, i.e., $\alpha(j \sep \theta) = \int\limits_{X^\theta} \sum\limits_{m \in \supp \widehat{\sigma}(\cdot \sep \theta, x)} \widehat{\sigma}(m \sep \theta,x) \widehat{\tau}(j \sep m) \,\mathrm{d}x \, \big/ \mu_0(\theta)$.

We derive a sharp characterization of equilibrium outcomes in the SMM game.

\begin{theorem}\label{thm:EQN-char-SMM}
    Let $\Theta$ be finite. Then $\alpha \in \Psi$ is an equilibrium outcome of the SMM game $\iff$ $\alpha$ is IC and obedient.
\end{theorem}

\begin{proof}
    $(\Longrightarrow)$ is proved exactly the same way as \cref{thm:outcome_chara} \ref{thm1a}. An equilibrium outcome must be IC or else S has a profitable deviation to fully revealing $x$ (which also reveals $\theta \in \Theta$ since $x \in X^\theta$).
    An equilibrium outcome must be obedient by Bayes' rule.

    $(\Longleftarrow)$ Consider an IC and obedient outcome $\alpha$. For every $\theta \in \Theta$, let $J^\theta := \supp \alpha(\cdot \sep \theta)$ be the set of actions that R takes with a positive probability when the realized state is $\theta$. Next, partition $X^\theta$ into a set of intervals $\{ X_j^\theta \}_{j \in J^\theta}$ such that $\lambda(X_j^\theta) \, / \, \lambda(X^\theta) = \alpha(j \sep \theta)$. Also, for each action $j \in J$, let $W_j := \bigcup\limits_{\theta \in \Theta} X_j^\theta$; by construction, $\{ W_j \}_{j\in J}$ is a partition of $[0,1]$.

    Now let S's strategy be $\widehat{\sigma}( m \sep \theta,x ) = \mathbbm{1}( m = W_j \text{ and } x \in W_j )$. Then R's posterior after an on-path message $W_j$ is $\widehat{q}(\theta \sep W_j) = \lambda(X_j^\theta) \, / \,\lambda(W_j)$. Furthermore, since $\alpha$ is obedient, for every action $j \in J$ such that $\lambda(W_j) > 0$, we have
    \begin{align*}
        \sum\limits_{\theta \in \Theta} \big( u(j,\theta) - u(j',\theta) \big) \alpha(j \sep \theta) \mu_0 (\theta)       & \geq 0
        \iff                                                                                                                                                                            \\
        \sum\limits_{\theta \in \Theta} \big( u(j,\theta) - u(j',\theta) \big) \frac{ \lambda(X_j^\theta) }{\lambda(W_j)} & \geq 0 \quad \text{for all } j' \in J,
    \end{align*}
    meaning that R prefers to take action $j$ after message $W_j$, so we let $\widehat{\tau}(j \sep W_j) = 1$. Off the path, let R be ``skeptical'' and assume that any unexpected message comes from the state in which S benefits from such deviation the most. Formally, for all $m \notin \{W_j\}_{j \in J}$, let $\widehat{q}(\cdot \sep m) \in \Delta A_{\underline{j}}$, where $\underline{j} \in J$ is the lowest action such that  $m \cap X^\theta \neq \varnothing$ and $\theta \in A_i$. Then playing action $\underline{j}$ is a best response to message $m$, so we let $\widehat{\tau}(\underline{j} \sep m) = 1$. Since $\alpha$ is IC, S does not have profitable deviations by the same argument as in the proof of \cref{thm:outcome_chara}. Deviations to on-path messages are not available because $\{ W_j \}_{j \in J}$ is a partition, and deviations to off-path messages are not profitable since the payoff from any deviation in state $\theta$ is at most $\underline{v}(\theta)$, which is below $v_\alpha(\theta)$ by the \eqref{IC-theta} constraint. Hence, $(\widehat{\sigma},\widehat{\tau},\widehat{q})$ is an equilibrium of the SMM game.
\end{proof}

In contrast to \cref{thm:outcome_chara}, IC and obedience are necessary and sufficient for an outcome to be an equilibrium outcome of the SMM game.
Two properties of the SMM game ensure that every IC and obedient outcome is an equilibrium outcome.
First, S's message space depends on $x$, which means S may receive different equilibrium payoffs in some state $\theta$ (but for different realizations of $x$).
Second, the message space is ``rich,'' meaning that for every vector $p =(p_1,\ldots,p_N) \in [0,1]^N$, there exists a message $m$ that is available in state $\theta \in \Theta$ with probability $p_\theta$.
This richness allows us to ``purify'' any nondeterministic outcome: the equilibria that we construct to implement an IC and obedient outcome is in pure strategies of both S and R.

Using the sharp equilibrium characterization of the SMM game, we derive the following results.

\begin{corollary}\label{cor:SMM-comm-outcome}
    Let $\Theta$ be finite. Then a commitment outcome is an equilibrium outcome of the SMM game if and only if it is IC.
\end{corollary}
\begin{corollary} \label{cor:SMM-binary}
    If $\Theta$ is finite and $|J| = 2$, then every commitment outcome is an equilibrium outcome of the SMM game.
\end{corollary}

\cref{cor:SMM-comm-outcome} is a direct consequence of \cref{thm:EQN-char-SMM}. \cref{cor:SMM-binary} follows from \cref{thm:EQN-char-SMM} and the fact that every commitment outcome is IC when R has two actions (see \citenop{Alonso2016} and our discussion after \cref{prop:2actionsIC}).

The set of equilibrium outcomes in the SMM game coincides with the set of IO outcomes found in KS if S's value function is quasiconvex in R's belief (KS Proposition 3), or when R chooses between two actions (KS Proposition 4). Generally, the set of IO outcomes is a subset of the set of equilibrium outcomes in KS, because IO imposes a stronger restriction on off-path beliefs than our equilibrium concept.

\section{Conclusion} \label{s: conclusion}

This paper examined a persuasion game with verifiable information in which a sender with transparent motives chooses which verifiable messages to send to a receiver in order to convince her to take a particular action from a finite set.
We showed that every equilibrium outcome must be incentive-compatible for the sender and obedient for the receiver. 
If an outcome is deterministic, then these conditions are both necessary and sufficient for it to be an equilibrium outcome. 
We also identified sufficient conditions under which the ex-ante commitment assumption in Bayesian persuasion can be replaced by communication with verifiable information. 
We showed that if the state space is rich, then a deterministic commitment outcome always exists; this commitment outcome is an equilibrium outcome if and only if the sender receives at least his complete information payoff in every state. 
If the receiver chooses between two actions, this condition is automatically satisfied.
We hope these results prove useful in applied settings.

\begin{singlespace}
    \setlength\bibitemsep{5pt}
    \printbibliography
\end{singlespace}
\end{document}

%% file: fig-v-underline.tex
\begin{tikzpicture}

\tikzset{
  font={\fontsize{11pt}{12}\selectfont}}

    \everymath\expandafter{\the\everymath\SetColor{black}}

    \draw[<->, thick] (3.5, 0) node[right]{\(\theta\)} -- (0,0) -- (0,3.25); 
    \draw[dashed] (0, 5/6) node[left]{\(1\)} -- (1, 5/6)-- (1, 0) node[below]{\(\frac{1}{3}\)};
    \draw[dashed] (0, 2.5) node[left]{\(3\)} -- (2, 2.5) -- (2, 0) node[below]{\(\frac{2}{3}\)};
    \draw[dashed] (3, 2.5) -- (3, 0) node[below]{\(1\)};
    \draw[line width = 0.75mm, MyRed] (0,0) node[below left, black]{\(0\)} -- (1,0);
    \draw[line width = 0.75mm, MyRed] (1,5/6) -- (2,5/6);
    \draw[ultra thick, MyRed, fill=white] (1,5/6) circle [radius=0.06];
    \draw[ultra thick, MyRed, fill=MyRed] (1,0) circle [radius=0.05];
    \draw[line width = 0.75mm, MyRed] (2,2.5) -- (3,2.5) node[right]{\(\underline{v}(\theta)\)};
    \draw[ultra thick, MyRed, fill=white] (2,2.5) circle [radius=0.06];
    \draw[ultra thick, MyRed, fill=MyRed] (2,5/6) circle [radius=0.05];
    \draw[ultra thick, MyRed, fill=MyRed] (3,2.5) circle [radius=0.05];
    \draw[ultra thick, MyRed, fill=MyRed] (0,0) circle [radius=0.05];
    
\end{tikzpicture}

%% file: fig-v-psi.tex
\begin{tikzpicture}

\tikzset{
  font={\fontsize{11pt}{12}\selectfont}}

    \everymath\expandafter{\the\everymath\SetColor{black}}

    \draw[<->, thick] (3.5, 0) node[right]{\(\theta\)} -- (0,0) -- (0,3.25); 

    \def\veight{0.5};
    \def\veleven{0.6875};
    \def\vtwentyone{1.3125};


    \draw[dashed] (0, 2.5) node[left]{\(3\)} -- (\veight,2.5) -- (\veight, 0) node[below,xshift=-0.75mm]{\(\frac{8}{48}\)};

    \draw[dashed] (\veleven,2.5) -- (\veleven,0) node[below,xshift=0.75mm]{\(\frac{11}{48}\)};

    \draw[dashed] (0, 5/6) node[left]{\(1\)} -- (\veleven,5/6);

    \draw[dashed] (0, 2.5)                   -- (\vtwentyone,2.5) -- (\vtwentyone, 0) node[below]{\(\frac{21}{48}\)};

    \draw[dashed] (0, 2.5)                   -- (3,2.5) -- (3, 0) node[below]{\(1\)};


    \draw[line width = 0.75mm, MyBlue] (0,0) -- (\veight,0);

    \draw[line width = 0.75mm, MyBlue] (\veleven,5/6) -- (\vtwentyone,5/6);

    \draw[line width = 0.75mm, MyBlue] (\veight,2.5) -- (\veleven,2.5);
    \draw[line width = 0.75mm, MyBlue] (\vtwentyone,2.5) -- (3,2.5) node[right] {$v_{\overline{\psi}} (\theta)$};
    

    \draw[ultra thick, MyBlue, fill=white] (\veight,0) circle [radius=0.06];

    \draw[ultra thick, MyBlue, fill=white] (\veleven,5/6) circle [radius=0.06];
    \draw[ultra thick, MyBlue, fill=white] (\vtwentyone,5/6) circle [radius=0.06];


    \draw[ultra thick, MyBlue, fill=MyBlue] (0,0) circle [radius=0.05];

    \draw[ultra thick, MyBlue, fill=MyBlue] (\veight,2.5) circle [radius=0.05];
    \draw[ultra thick, MyBlue, fill=MyBlue] (\veleven,2.5) circle [radius=0.05];

    \draw[ultra thick, MyBlue, fill=MyBlue] (\vtwentyone,2.5) circle [radius=0.05];
    \draw[ultra thick, MyBlue, fill=MyBlue] (3,2.5) circle [radius=0.05];
    
\end{tikzpicture}

%% file: fig-both.tex
\begin{tikzpicture}

\tikzset{
  font={\fontsize{11pt}{12}\selectfont}}

    \everymath\expandafter{\the\everymath\SetColor{black}}

    \draw[<->, thick] (3.5, 0) node[right]{\(\theta\)} -- (0,0) -- (0,3.25); 

    \def\veight{0.5};
    \def\veleven{0.6875};
    \def\vtwentyone{1.3125};

    \def\labelshift{-0.75};


    \draw[dashed] (0, 2.5) node[left]{\(3\)} -- (\veight,2.5) -- (\veight, 0) node[below,yshift=\labelshift]{\(\frac{8}{48}\)};

    \draw[dashed] (\veleven,2.5) -- (\veleven,0);

    \draw[dashed] (0, 5/6) node[left]{\(1\)} -- (\veleven,5/6);

    \draw[dashed] (0, 2.5)                   -- (\vtwentyone,2.5) -- (\vtwentyone, 0) node[below,yshift=\labelshift]{\(\frac{21}{48}\)};

    \draw[dashed] (0, 2.5)                   -- (3,2.5) -- (3, 0) node[below]{\(1\)};

    \draw[dashed] (1, 5/6)-- (1, 0) node[below,yshift=\labelshift]{\(\frac{1}{3}\)};
    \draw[dashed] (2, 2.5) -- (2, 0) node[below,yshift=\labelshift]{\(\frac{2}{3}\)};


    \def\yshift{-0.5mm};
    
    \draw[line width = 0.75mm, MyRed, yshift = \yshift] (0,0) node[below left, black]{\(0\)} -- (1,0);
    
    \draw[line width = 0.75mm, MyRed, yshift = \yshift] (1,5/6) -- (2,5/6);

    \draw[line width = 0.75mm, MyRed, yshift = \yshift] (2,2.5) -- (3,2.5) node[right,yshift = -2.25mm]{\(\underline{v}(\theta)\)};

    \draw[ultra thick, MyRed, fill=white, yshift = \yshift] (1,5/6) circle [radius=0.06];
    
    \draw[ultra thick, MyRed, fill=white, yshift = \yshift] (2,2.5) circle [radius=0.06];

    \draw[ultra thick, MyRed, fill=MyRed, yshift = \yshift] (2,5/6) circle [radius=0.05];
    
    \draw[ultra thick, MyRed, fill=MyRed, yshift = \yshift] (3,2.5) circle [radius=0.05];
    
    \draw[ultra thick, MyRed, fill=MyRed, yshift = \yshift] (0,0) circle [radius=0.05];

    \draw[ultra thick, MyRed, fill=MyRed, yshift = \yshift] (1,0) circle [radius=0.05];


    \def\yshift{0.25mm};


    \draw[line width = 0.75mm, MyBlue, yshift = \yshift] (0,0) -- (\veight,0);

    \draw[line width = 0.75mm, MyBlue, yshift = \yshift] (\veleven,5/6) -- (\vtwentyone,5/6);

    \draw[line width = 0.75mm, MyBlue, yshift = \yshift] (\veight,2.5) -- (\veleven,2.5);
    \draw[line width = 0.75mm, MyBlue, yshift = \yshift] (\vtwentyone,2.5) -- (3,2.5) node[right,yshift = 2.25mm] {$v_{\overline{\psi}} (\theta)$};
    

    \draw[ultra thick, MyBlue, fill=white, yshift = \yshift] (\veight,0) circle [radius=0.06];

    \draw[ultra thick, MyBlue, fill=white, yshift = \yshift] (\veleven,5/6) circle [radius=0.06];
    \draw[ultra thick, MyBlue, fill=white, yshift = \yshift] (\vtwentyone,5/6) circle [radius=0.06];


    \draw[ultra thick, MyBlue, fill=MyBlue, yshift = \yshift] (0,0) circle [radius=0.05];

    \draw[ultra thick, MyBlue, fill=MyBlue, yshift = \yshift] (\veight,2.5) circle [radius=0.05];
    \draw[ultra thick, MyBlue, fill=MyBlue, yshift = \yshift] (\veleven,2.5) circle [radius=0.05];

    \draw[ultra thick, MyBlue, fill=MyBlue, yshift = \yshift] (\vtwentyone,2.5) circle [radius=0.05];
    \draw[ultra thick, MyBlue, fill=MyBlue, yshift = \yshift] (3,2.5) circle [radius=0.05];

\end{tikzpicture}

%% file: library.bib
@article{KoesslerRenault2012,
   title = {When {Does} a {Firm} {Disclose} {Product} {Information}?},
   volume = {43},
   issn = {0741-6261},
   url = {https://www.jstor.org/stable/41723348},
   number = {4},
   urldate = {2023-11-10},
   journal = {The RAND Journal of Economics},
   author = {Koessler, Frédéric and Renault, Régis},
   year = {2012},
   pages = {630--649},
}

@article{Zapechelnyuk2023,
   title = {On the {Equivalence} of {Information} {Design} by {Uninformed} and {Informed} {Principals}},
   volume = {76},
   issn = {1432-0479},
   url = {https://doi.org/10.1007/s00199-023-01495-z},
   doi = {10.1007/s00199-023-01495-z},
   number = {4},
   urldate = {2023-11-08},
   journal = {Economic Theory},
   author = {Zapechelnyuk, Andriy},
   month = nov,
   year = {2023},
   pages = {1051--1067},
}

@article{KoesslerSkreta2023,
   title = {Informed Information Design},
   issn = {0022-3808},
   url = {https://www.journals.uchicago.edu/doi/abs/10.1086/724843},
   doi = {10.1086/724843},
   journaltitle = {Journal of Political Economy},
   author = {Koessler, Fr{\'e}d{\'e}ric and Skreta, Vasiliki},
   urldate = {2023-07-06},
   date = {2023-03-28},
}

@article{GlazerRubinstein2006,
   author = {Jacob Glazer and Ariel Rubinstein},
   doi = {j.ctt46mwfz.9},
   journal = {Theoretical Economics},
   pages = {395-410},
   title = {A Study in the Pragmatics of Persuasion: a Game Theoretical Approach},
   volume = {1},
   url = {http://econtheory.org.},
   year = {2006},
}

@article{GlazerRubinstein2004,
   author = {Jacob Glazer and Ariel Rubinstein},
   doi = {10.1111/j.1468-0262.2004.00551.x},
   issn = {0012-9682},
   issue = {6},
   journal = {Econometrica},
   month = {11},
   pages = {1715--1736},
   title = {On Optimal Rules of Persuasion},
   volume = {72},
   url = {http://doi.wiley.com/10.1111/j.1468-0262.2004.00551.x},
   year = {2004},
}

@article{HartEtAl2017,
   author = {Sergiu Hart and Ilan Kremer and Motty Perry},
   doi = {10.1257/aer.20150913},
   issn = {0002-8282},
   issue = {3},
   journal = {American Economic Review},
   month = {3},
   pages = {690--713},
   publisher = {American Economic Association},
   title = {Evidence Games: Truth and Commitment},
   volume = {107},
   url = {https://pubs.aeaweb.org/doi/10.1257/aer.20150913},
   year = {2017},
}

@article{Sher2011,
   author = {Itai Sher},
   doi = {10.1016/j.geb.2010.05.008},
   issn = {08998256},
   issue = {2},
   journal = {Games and Economic Behavior},
   month = {3},
   pages = {409--419},
   title = {Credibility and Determinism in a Game of Persuasion},
   volume = {71},
   url = {https://linkinghub.elsevier.com/retrieve/pii/S0899825610000953},
   year = {2011},
}

@article{DvoretzkyEtAl1951,
   author = {A. Dvoretzky and A. Wald and J. Wolfowitz},
   doi = {10.1214/aoms/1177729689},
   issn = {0003-4851},
   issue = {1},
   journal = {The Annals of Mathematical Statistics},
   month = {3},
   pages = {1-21},
   title = {Elimination of Randomization in Certain Statistical Decision Procedures and Zero-Sum Two-Person Games},
   volume = {22},
   url = {http://projecteuclid.org/euclid.aoms/1177729689},
   year = {1951},
}

@article{Bardhi2018,
    author       = {Bardhi, Arjada and Guo, Yingni},
    url          = {https://econtheory.org/ojs/index.php/te/article/view/2834},
    date         = {2018},
    doi          = {10.3982/TE2834},
    issn         = {1933-6837},
    journaltitle = {Theoretical Economics},
    number       = {3},
    pages        = {1111--1149},
    title        = {{Modes of Persuasion toward Unanimous Consent}},
    volume       = {13},
}

@article{Milgrom1986,
    author       = {Milgrom, Paul R. and Roberts, John},
    url          = {http://doi.wiley.com/10.2307/2555625},
    date         = {1986},
    doi          = {10.2307/2555625},
    issn         = {07416261},
    journaltitle = {The RAND Journal of Economics},
    number       = {1},
    pages        = {18},
    title        = {{Relying on the Information of Interested Parties}},
    volume       = {17},
}

@article{Chakraborty2010,
    author       = {Chakraborty, Archishman and Harbaugh, Rick},
    date         = {2010},
    doi          = {10.1257/aer.100.5.2361},
    issn         = {00028282},
    journaltitle = {American Economic Review},
    number       = {5},
    pages        = {2361--2382},
    title        = {{Persuasion by {Cheap} {Talk}}},
    volume       = {100},
}

@article{Ostrovsky2010,
    author       = {Ostrovsky, Michael and Schwarz, Michael},
    url          = {https://pubs.aeaweb.org/doi/10.1257/mic.2.2.34},
    date         = {2010-05},
    doi          = {10.1257/mic.2.2.34},
    issn         = {1945-7669},
    journaltitle = {American Economic Journal: Microeconomics},
    number       = {2},
    pages        = {34--63},
    title        = {{Information Disclosure and Unraveling in Matching Markets}},
    volume       = {2},
}

@article{Lipnowski2020,
    author       = {Lipnowski, Elliot and Ravid, Doron},
    url          = {https://www.econometricsociety.org/doi/10.3982/ECTA15674},
    date         = {2020},
    doi          = {10.3982/ECTA15674},
    issn         = {0012-9682},
    journaltitle = {Econometrica},
    number       = {4},
    pages        = {1631--1660},
    title        = {{Cheap Talk With Transparent Motives}},
    volume       = {88},
}

@article{Milgrom2008,
    author       = {Milgrom, Paul R.},
    url          = {https://voprecotest.elpub.ru/jour/article/view/851 https://pubs.aeaweb.org/doi/10.1257/jep.22.2.115},
    date         = {2008-03},
    doi          = {10.1257/jep.22.2.115},
    issn         = {0895-3309},
    journaltitle = {Journal of Economic Perspectives},
    number       = {2},
    pages        = {115--131},
    title        = {{What the Seller Won't Tell You: Persuasion and Disclosure in Markets}},
    volume       = {22},
}

@article{Kamenica2011,
    author       = {Kamenica, Emir and Gentzkow, Matthew},
    url          = {http://www.ncbi.nlm.nih.gov/pubmed/10525498 http://pubs.aeaweb.org/doi/10.1257/aer.101.6.2590},
    date         = {2011-10},
    doi          = {10.1257/aer.101.6.2590},
    issn         = {0002-8282},
    journaltitle = {American Economic Review},
    number       = {6},
    pages        = {2590--2615},
    title        = {{Bayesian Persuasion}},
    volume       = {101},
}

@article{Alonso2016,
    author       = {Alonso, Ricardo and Câmara, Odilon},
    url          = {http://pubs.aeaweb.org/doi/10.1257/aer.20140737},
    date         = {2016-11},
    doi          = {10.1257/aer.20140737},
    issn         = {0002-8282},
    journaltitle = {American Economic Review},
    number       = {11},
    pages        = {3590--3605},
    title        = {{Persuading Voters}},
    volume       = {106},
}

@unpublished{Lipnowski2020b,
    author       = {Lipnowski, Elliot},
    date         = {2020},
    note = {working paper},
    title        = {{Equivalence of Cheap Talk and Bayesian Persuasion in a Finite Continuous Model}},
}

@article{Gentzkow2016,
    author       = {Gentzkow, Matthew and Kamenica, Emir},
    url          = {https://pubs.aeaweb.org/doi/10.1257/aer.p20161049},
    date         = {2016-05},
    doi          = {10.1257/aer.p20161049},
    issn         = {0002-8282},
    journaltitle = {American Economic Review},
    number       = {5},
    pages        = {597--601},
    title        = {{A Rothschild-Stiglitz Approach to Bayesian Persuasion}},
    volume       = {106},
}

@article{Kolotilin2015,
    author       = {Kolotilin, Anton},
    publisher    = {Elsevier Inc.},
    url          = {http://dx.doi.org/10.1016/j.geb.2015.02.006 https://linkinghub.elsevier.com/retrieve/pii/S0899825615000287},
    date         = {2015-03},
    doi          = {10.1016/j.geb.2015.02.006},
    issn         = {08998256},
    journaltitle = {Games and Economic Behavior},
    pages        = {215--226},
    title        = {{Experimental Design to Persuade}},
    volume       = {90},
}

@article{BullWatson2007,
    author       = {Bull, Jesse and Watson, Joel},
    date         = {2007},
    doi          = {10.1016/j.geb.2006.03.003},
    issn         = {08998256},
    journaltitle = {Games and Economic Behavior},
    number       = {1},
    pages        = {75--93},
    title        = {{Hard Evidence and Mechanism Design}},
    volume       = {58},
}

@article{Gehlbach2014,
    author       = {Gehlbach, Scott and Sonin, Konstantin},
    publisher    = {Elsevier B.V.},
    url          = {http://dx.doi.org/10.1016/j.jpubeco.2014.06.004 https://linkinghub.elsevier.com/retrieve/pii/S0047272714001443},
    date         = {2014-10},
    doi          = {10.1016/j.jpubeco.2014.06.004},
    issn         = {00472727},
    journaltitle = {Journal of Public Economics},
    pages        = {163--171},
    title        = {{Government Control of the Media}},
    volume       = {118},
}

@article{Grossman1981,
    author       = {Grossman, Sanford J.},
    url          = {http://www.journals.uchicago.edu/doi/10.1086/466995 https://www.journals.uchicago.edu/doi/10.1086/466995},
    date         = {1981-12},
    doi          = {10.1086/466995},
    issn         = {0022-2186},
    journaltitle = {The Journal of Law and Economics},
    number       = {3},
    pages        = {461--483},
    title        = {{The Informational Role of Warranties and Private Disclosure about Product Quality}},
    volume       = {24},
}

@article{Milgrom1981,
    author       = {Milgrom, Paul R.},
    url          = {https://www.jstor.org/stable/3003562?origin=crossref http://www.jstor.org/stable/3003562?origin=crossref},
    date         = {1981},
    doi          = {10.2307/3003562},
    issn         = {0361915X},
    journaltitle = {The Bell Journal of Economics},
    number       = {2},
    pages        = {380},
    title        = {{Good News and Bad News: Representation Theorems and Applications}},
    volume       = {12},
}

@article{Ben-Porath2019,
    author       = {Ben-Porath, Elchanan and Dekel, Eddie and Lipman, Barton L.},
    url          = {https://www.econometricsociety.org/doi/10.3982/ECTA14991},
    date         = {2019},
    doi          = {10.3982/ECTA14991},
    issn         = {0012-9682},
    journaltitle = {Econometrica},
    number       = {2},
    pages        = {529--566},
    title        = {{Mechanisms With Evidence: Commitment and Robustness}},
    volume       = {87},
}

@article{Boleslavsky2015,
    author       = {Boleslavsky, Raphael and Cotton, Christopher},
    url          = {https://pubs.aeaweb.org/doi/10.1257/mic.20130080},
    date         = {2015-05},
    doi          = {10.1257/mic.20130080},
    issn         = {1945-7669},
    journaltitle = {American Economic Journal: Microeconomics},
    number       = {2},
    pages        = {248--279},
    title        = {{Grading Standards and Education Quality}},
    volume       = {7},
}

@unpublished{Zhang2022,
      title={Withholding Verifiable Information}, 
      author={Kun Zhang},
      year={2022},
      note={working paper}, 
}

@unpublished{GieczewskiTitova,
      title={Coalition-Proof Disclosure}, 
      author={Germ{\'a}n Gieczewski and Maria Titova},
      year={2024},
      note={working paper}, 
}

@article{Grossman1980,
	author = {Sanford J. Grossman and Oliver D. Hart},
	journal = {The Journal of Finance},
	number = {2},
	pages = {323--334},
	title = {Disclosure Laws and Takeover Bids},
	volume = {35},
	year = {1980}}

@article{Dranove2010,
	author = {David Dranove and Ginger Zhe Jin},
	journal = {Journal of Economic Literature},
	number = {4},
	pages = {935--963},
	title = {Quality Disclosure and Certification: Theory and Practice},
	volume = {48},
	year = {2010}}

@unpublished{Ali2024,
  title={From Design to Disclosure},
  author={S. Nageeb Ali and Andreas Kleiner and Kun Zhang},
  note={working paper},
  year={2024}
}

@article{PerezRichet2014,
  title={Interim Bayesian Persuasion: First Steps},
  author={Perez-Richet, Eduardo},
  journal={American Economic Review},
  volume={104},
  number={5},
  pages={469--474},
  year={2014}
}

@article{Terstiege2023,
  title={Experiments Versus Distributions of Posteriors},
  author={Terstiege, Stefan and Wasser, C{\'e}dric},
  journal={Mathematical Social Sciences},
  volume={125},
  pages={58--60},
  year={2023}
}

@article{Dye1985,
  title={{Disclosure of Nonproprietary Iinformation}},
  author={Dye, Ronald A},
  journal={{Journal of Accounting Research}},
  pages={123--145},
  year={1985}
}

@article{KamenicaLin2024,
  title={Commitment and Randomization in Communication},
  author={Kamenica, Emir and Lin, Xiao},
  journal={arXiv preprint arXiv:2410.17503},
  year={2024}
}
